\documentclass[a4paper,12pt]{article}
\usepackage{bm}
\usepackage{hyperref}
\usepackage{comment}
\usepackage{enumerate}
\usepackage{here}
\usepackage{latexsym}
\usepackage{array}
\usepackage{mathrsfs}
\usepackage{geometry}
\usepackage{fancyhdr}
\renewcommand{\title}[1]{

\begin{center} \Large \bf #1 \end{center}
}

\renewcommand{\author}[2]{
 \begin{center} #1  \vspace{3mm} \\
  #2 \\
 \end{center}
\addvspace{\baselineskip}
}

\usepackage{amssymb}
\usepackage{amsmath}

\usepackage{amsthm}
\newtheorem{theorem}{Theorem}[section]
\newtheorem{proposition}[theorem]{Proposition}

\theoremstyle{definition}

\theoremstyle{remark}

\makeatletter
\@addtoreset{equation}{section}

\makeatother

\begin{document}
\baselineskip 5mm
\title{ Real symmetric $ \Phi^4$-matrix model \\
 as Calogero-Moser model }
\author{${}^1$Harald Grosse,  ${}^2$Naoyuki Kanomata, ${}^2$Akifumi Sako, ${}^3$Raimar Wulkenhaar
}
{${}^{1,2,3}$
Erwin Schr\"odinger International Institute for Mathematics and Physics, \\
University of Vienna, Boltzmanngasse 9, 1090 Vienna, Austria \vspace{3mm}\\

${}^1$
Faculty of Physics, University of Vienna, Boltzmanngasse 5, 
1090 Vienna, Austria
\vspace{3mm}\\

${}^2$
Tokyo University of Science, 1-3 Kagurazaka, Shinjuku-ku, Tokyo, 162-8601, Japan
\vspace{3mm}\\

${}^3$
Mathematisches Institut, Universit\"at M\"unster,
Einsteinstra{\ss}e 62, D-48149 M\"unster, Germany
}
\noindent
\vspace{1cm}

\abstract{ 
We study a real symmetric $\Phi^4$-matrix model whose kinetic term 
is given by $\mathrm{Tr}( E \Phi^2)$, where $E$ is a positive diagonal matrix 
without degenerate eigenvalues.
We show that the partition function of this matrix model corresponds
to a zero-energy solution of a Sch\"odinger type
equation with Calogero-Moser Hamiltonian.
A family of differential equations satisfied by the partition
function is also  obtained from the Virasoro algebra.
}
%
%
%
\section{Introduction}\label{sect1}
It has recently been shown 
that the partition function of a certain Hermitian $\Phi^4$-matrix model
corresponds to a zero-energy solution of a Schr\"odinger equation
for the Hamiltonian of $N$-body harmonic oscillator system
\cite{Grosse:2023jcb}.
This $\Phi^4$-matrix model is obtained 
by changing the potential of the Kontsevich model 
\cite{Kontsevich:1992ti} from $\Phi^3$ to $\Phi^4$.
\footnote{It was introduced while studying a scalar field defined 
on a noncommutative space-time 
and studied over years \cite{Grosse:2005ig,Grosse:2006qv,Grosse:2006tc}
to resolve the IR/UV-mixing problem. 
Recent developments are summarized in \cite{Branahl:2021slr}. }
The $N$-body harmonic oscillator system can be 
extended to the integrable Calogero-Moser model
\cite{Calogero:1970nt,Moser}. It is thus natural
to conjecture that there should be matrix models whose 
partition functions satisfy the Schr\"odinger equation for the
Calogero-Moser model. 
It is precisely this which we demonstrate in this paper.

\bigskip

Let $\Phi$ be a \emph{real symmetric} $N\times N$ matrix,
$E$ be a positive diagonal $N\times N$ matrix 
$E := \mathrm{diag} (E_1, E_2 , \cdots ,E_N )$ without
degenerate eigenvalues,
and $\eta$ be a positive real number, called coupling constant.
We deal in this paper with the following symmetric one-matrix model defined by
\begin{align}
S_E&= N~ \mathrm{Tr} \{ E \Phi^2 + \frac{\eta}{4} \Phi^4  \}
\notag \\
&= N \left( 
\sum_{i,j}^N  E_{i}\Phi_{ij}\Phi_{ji}
+ \frac{\eta}{4} \sum_{i,j,k,l}^N
\Phi_{ij}\Phi_{jk}\Phi_{kl}\Phi_{li}
\right). \label{act1}
\end{align} 

The main theorem of this paper is:
\begin{theorem}\label{main_thm1}
Let $Z(E, \eta)$ be the partition function defined by
\begin{equation}
Z(E, \eta)= \int_{S_N} d\Phi ~e^{-S_E[\Phi]} ,
\label{act1a}
\end{equation}
where $S_N$ is the space of real symmetric $N\times N$-matrices.
Let $\Delta(E)$ be the Vandermonde  determinant 
$\Delta (E) := \prod_{k<l} (E_l -E_k)$.
Then the function
\begin{align*}
\Psi (E, \eta ) := 
e^{-\frac{N}{\eta}\sum_{i=1}^{N}E_{i}^{2}}\Delta(E)^{\frac{1}{2}}
{Z}(E,\eta)
\end{align*}
is a zero-energy solution of the Schr\"odinger type equation
\begin{align*}
{\mathcal H}_{CM} \Psi (E, \eta ) = 0, 
\end{align*}
where ${\mathcal H}_{CM}$ is the Hamiltonian for the Calogero-Moser model:
\begin{align}\label{Calogero_H_for_Matrix}
{\mathcal H}_{CM}:=\frac{-\eta}{2N}\left(\sum_{i=1}^{N}\frac{\partial^{2}}{\partial E_{i}^{2}}+\frac{1}{4}\sum_{i\neq j}\frac{1}{(E_{i}-E_{j})^{2}}\right)+2\frac{N}{\eta}\sum_{i=1}^{N}E_{i}^{2}.
\end{align}
In this sense, this matrix model is a solvable system.
\end{theorem}

Furthermore, since the Calogero-Moser model admits a 
Virasoro algebra representation, 
it gives rise to a family of differential equations satisfied by the 
partition function ${Z}(E,\eta)$.
We will see this result in Theorem \ref{main2}.

\section{Schwinger-Dyson equation}\label{sect2}

Let $\Phi$ be a real symmetric $N\times N$ matrix.
Let $H$ be a real symmetric $N\times N$ matrix
with nondegenerate eigenvalues
$\{E_1, E_2 , \cdots ,E_N ~ | ~ E_i \neq E_j ~\mbox{for}~ i \neq j \}$.
Let $\eta$ be a real positive number.
We consider the following action
\begin{align}
S&= N~ \mathrm{Tr} \{ H \Phi^2 + \frac{\eta}{4} \Phi^4  \}
\notag \\
&= N \left( 
\sum_{i,j,k}^N  H_{ij}\Phi_{jk}\Phi_{ki}
+ \frac{\eta}{4} \sum_{i,j,k,l}^N
\Phi_{ij}\Phi_{jk}\Phi_{kl}\Phi_{li}
\right).
\label{action_S}
\end{align}
The partition function is defined by
\begin{align}
Z(E, \eta) := \int_{S_N} d \Phi ~e^{-S} ,
\label{partitionfunction}
\end{align}
where 
$\displaystyle{d}\Phi=\prod_{i=1}^{N}d\Phi_{ii}\prod_{1\leq i<j\leq N}d\Phi_{ij}$
is the Lebesgue measure 
and $S_N$ the space of real symmetric $N\times N$ matrices.
We denote expectation values with this action $S$
by 
$\displaystyle \langle O \rangle := \int_{S_N} d \Phi ~ O\; e^{-S} $.
Note that we do not normalize it here, i.e.\
$\langle 1 \rangle = Z(E, \eta) \neq 1$.
Note that the partition function $Z(E, \eta) $ depends only on the
eigenvalues of $H$, because the integral measure is $O(N)$ invariant.
Indeed $Z(E, \eta) $ is equal to the partition function (\ref{act1a})
built from the action $S_E$ in (\ref{act1}).

The following discussion in this section runs parallel to 
\cite{Grosse:2023jcb}, so the calculations in \cite{Grosse:2023jcb} will also be helpful.

\bigskip

First, a Schwinger-Dyson equation is derived from 
\begin{align}
\int_{S_N} {d}\Phi \frac{\partial}{\partial\Phi_{tt}}\left(\Phi_{tt}e^{-S[\Phi]}\right)=0\notag,
\end{align}
which is expressed as
\begin{align}
{Z}(E,\eta)-2N\sum_{i=1}^{N}\left<H_{it}\Phi_{tt}\Phi_{ti}\right>-\eta N\sum_{k,l=1}^{N}\left<\Phi_{tk}\Phi_{kl}\Phi_{lt}\Phi_{tt}\right>=0. \label{a}
\end{align}
Similarly, for $p\neq s$, from
\begin{align}
\int_{S_N} {d}\Phi \frac{\partial}{\partial\Phi_{ps}}\left(\Phi_{ps}e^{-S[\Phi]}\right)=0 , 
\end{align}
the following is obtained:
\begin{align}
{Z}(E,\eta)-2N\sum_{i=1}^{N}\left(\left<H_{ip}\Phi_{ps}\Phi_{si}\right>+\left<H_{si}\Phi_{ip}\Phi_{ps}\right>\right)-2N\eta\sum_{k,l=1}^{N}\left<\Phi_{sk}\Phi_{kl}\Phi_{lp}\Phi_{ps}\right>=0. 
\label{b}
\end{align}
From (\ref{a}) and (\ref{b}), after taking sum over the indices
$t,p,s$,  we get the follwing:
\begin{align}
\frac{N(N+1)}{2}{Z}(E,\eta)-2N\sum_{i,p,s=1}^{N}H_{ip} \left<\Phi_{is}\Phi_{sp}\right>-\eta N\sum_{k,l,s,p=1}^{N}\left<\Phi_{ps}\Phi_{sk}\Phi_{kl}\Phi_{lp}\right>=0. \label{500}
\end{align}
By using
\begin{align}
&\frac{\partial{Z}(E,\eta)}{\partial H_{ps}}=-2N\sum_{k=1}^{N}\left<\Phi_{pk}\Phi_{ks}\right>\hspace{2mm}\mbox{for}\hspace{2mm}\ p\neq s
\notag \\
&\frac{\partial{Z}(E,\eta)}{\partial H_{pp}}=-N\sum_{k=1}^{N}\left<\Phi_{pk}\Phi_{kp}\right>
\notag \\
&\frac{\partial^{2}{Z}(E,\eta)}{\partial H_{ps}\partial H_{tu}}=4N^{2}\sum_{k,l=1}^{N}\left<\Phi_{pk}\Phi_{ks}\Phi_{tl}\Phi_{lu}\right> \mbox{for}\hspace{2mm} p\neq s, t\neq u
\notag \\
&\frac{\partial^{2}{Z}(E,\eta)}{\partial H_{pp}\partial H_{pp}}=N^{2}\sum_{k,l=1}^{N}\left<\Phi_{pk}\Phi_{kp}\Phi_{pl}\Phi_{lp}\right>, \notag
\end{align}
a partial differential equation is obtained:
\begin{align}\label{2_7}
&\frac{N(N+1)}{2}{Z}(E,\eta)+ \sum_{ i\neq p} H_{ip}\frac{\partial}{\partial H_{ip}}{Z}(E,\eta)+2\sum_{p=1}^{N}H_{pp}\frac{\partial}{\partial H_{pp}}{Z}(E,\eta)\notag\\
&-\frac{\eta}{N}\sum_{s=1}^{N}\frac{\partial^{2}}{\partial H_{ss}\partial H_{ss}}{Z}(E,\eta)-\frac{\eta}{4N}\sum_{s\neq l}
\frac{\partial^{2}}{\partial H_{sl}\partial H_{ls}}{Z}(E,\eta)=0,
\end{align}
where we denote $ \sum_{p=1}^{N}\sum_{i=1, i\neq p}^{N}$ by $\displaystyle \sum_{i \neq p} $.
We define $H'_{ij}$ by $H_{ii}=\sqrt{2}H^{'}_{ii} $ 
for $i=1,\cdots ,N$
and $H_{ij}=H^{'}_{ij}$ for $i,j=1,\cdots,N\hspace{2mm}(i\neq j)$, and 
we use an indices set $U=\{(p,s)|\hspace{2mm}p\leq s,\hspace{2mm}p,s\in\{1,2,\cdots, N\}\}$,
for convenience.

\begin{proposition}\label{prop2_1}
The partition function $Z(E, \eta)$ satisfies
the following partial differential equation:
\begin{align}
{\mathcal L}_{SD}^H Z(E, \eta) = 0 . \label{SD_H}
\end{align}
Here, ${\mathcal L}_{SD}^H $ is a second order differential operator
defined by
\begin{align}
-{\mathcal L}_{SD}^H:=&\frac{N(N+1)}{2}+2\sum_{(p,s)\in U} H_{ps}\frac{\partial}{\partial H_{ps}}
-\frac{\eta}{2N} \sum_{(p,s)\in U} \frac{\partial^{2}}{\partial H^{'}_{ps}\partial H^{'}_{sp}} .
\label{LSD}
\end{align}
\end{proposition}

Next we rewrite this Schwinger-Dyson equation in terms of the eigenvalues
$E_n (n= 1,2, \cdots , N)$ of $H$. References \cite{Itzykson:1992ya,Kimura}
are helpful in the following calculations.
Let $P(x)$ be the characteristic polynomial: 
$$P(x): = \det (x~ Id_N - H) =\det B=\prod_{i=1}^N (x-E_i),$$
where $B(x)=x~ Id_N - H$. Using this $P(x)$, the formula 
\begin{align}
\frac{\partial E_{t}}{\partial H_{ij}}=&\frac{2(^T\!\widetilde{B}(E_{t}))_{ij}-(^T\!\widetilde{B}(E_{t}))_{ii}\delta_{ij}}{P'(E_{t})}
\label{formulaB}
\end{align}
for the derivative
is obtained, where $\widetilde{B}$ is the cofactor matrix of $B$.
The proof of (\ref{formulaB}) is given in Appendix \ref{appendixA}.

At first, let us rewrite the second and the third terms of (\ref{2_7})
by using (\ref{formulaB}).
Since  $\widetilde{B}$ is a symmetric matrix,
\begin{align}
&2\sum_{(p,s)\in U}^{N}H_{ps}\frac{\partial}{\partial H_{ps}}{Z}(E,\eta)=2\sum_{(p,s)\in U}^{N}\sum_{k=1}^{N}H_{ps}\frac{2(\widetilde{B}(E_{k}))_{ps}-(\widetilde{B}(E_{k}))_{pp}\delta_{ps}}{P'(E_{k})}\frac{\partial Z}{\partial E_{k}}\notag\\
&=2\sum_{p,k,s=1}^{N}H_{ps}\frac{\widetilde{B}(E_{k})_{ps}}{P'(E_{k})}\frac{\partial Z}{\partial E_{k}}\notag\\
&=-2\sum_{p,k,s=1}^{N}(E_{k}\delta_{ps}-H_{ps})\frac{\widetilde{B}(E_{k})_{ps}}{P'(E_{k})}\frac{\partial Z}{\partial E_{k}}
+2\sum_{p,k,s=1}^{N}E_{k}\delta_{ps}\frac{\widetilde{B}(E_{k})_{ps}}{P'(E_{k})}\frac{\partial Z}{\partial E_{k}}. \notag
\end{align}
Due to the fact that
\begin{align}
\sum_{s=1}^{N}(E_{k}\delta_{ps}-H_{ps}){\widetilde{B}(E_{k})_{ps}}
= \det B(E_{k}) = P(E_k)=0
\end{align}
and 
\begin{align}
\sum_{p,s=1}^{N}\delta_{ps}{\widetilde{B}(E_{k})_{ps}}
=
\sum_{p=1}^{N}\det
\begin{pmatrix}
E_{k}-H_{11}&-H_{12}&\cdots&\cdots&-H_{1N}\\\vdots&\ddots&   &  &\vdots\\
0&\cdots&\delta_{pp}&\cdots&0\\
\vdots&&  & \ddots&\vdots\\
-H_{N1}&-H_{N2}&\cdots&\cdots&E_{k}-H_{NN}
\end{pmatrix}
= P'(E_k),
\end{align}
we finally get 
\begin{align}
2\sum_{(p,s)\in U}^{N}H_{ps}\frac{\partial}{\partial H_{ps}}{Z}(E,\eta)
=&2\sum_{k=1}^{N}E_{k}\frac{\partial Z}{\partial E_{k}} . \label{106}
\end{align}
As a next step, we rewrite the Laplacian
$\displaystyle \sum_{(p,s)\in U} \frac{\partial^{2}}{\partial H^{'}_{ps}\partial H^{'}_{sp}}{Z}
$
in terms of $E_p$.
It is a well-known fact (see e.g.\ \cite[sec.\ 1.2]{Eynard:2015aea})
that in terms of the Vandermonde  determinant 
$\Delta (E) := \prod_{k<l} (E_l -E_k)$,
the Jacobian for the change of variables reads
\begin{align}
\prod_{i=1}^{N}dH_{ii}\!\!\!\!\!\!
\prod_{1\leq i<j\leq N}
\!\!\!\!\!\!
dH_{ij}
=\Delta(E)\prod_{i=1}^{N}dE_{i}
\!\!\!\!\!\!
\prod_{1\leq k<l\leq N}
\!\!\!\!\!\!
dO_{lk}
=(\sqrt{2})^{N}\prod_{i=1}^{N}dH'_{ii}
\!\!\!\!\!\!
\prod_{1\leq i<j\leq N}
\!\!\!\!\!\!
dH^{'}_{ij},
\end{align}
where $\prod_{1\leq k<l\leq N}dO_{lk}$ is the Haar measure on $O(n)$.
Then the Laplacian is rewritten as
\begin{align}
 \sum_{(p,s)\in U} \frac{\partial^{2}}{\partial H^{'}_{ps}\partial H^{'}_{sp}}{Z}(E,\eta)
=&\frac{(\sqrt{2})^{N}}{\Delta(E)} \sum_{i=1}^{N}\frac{\partial}{\partial E_{i}}\left(\frac{\displaystyle\Delta(E)}{(\sqrt{2})^{N}}\frac{\partial}{\partial E_{i}}\right){Z}(E,\eta)\notag\\
=&
\sum_{l\neq i}^{N}
\frac{1}{E_{i}-E_{l}}\frac{\partial}{\partial E_{i}}{Z}(E,\eta)+\sum_{i=1}^{N}\frac{\partial^{2}}{\partial E_{i}^{2}}{Z}(E,\eta) .
\label{107}
\end{align}
From (\ref{106}), (\ref{107}) and Proposition \ref{prop2_1},
we obtain the following.
\begin{theorem}\label{thm2_2}
The partition function defined by ${Z}(E,\eta):=\int_{S_N} d \Phi\exp\left(-S[\Phi]\right)$
satisfies the partial differential equation
\begin{align}
{\mathcal L}_{SD} Z(E, \eta) = 0 ,
\label{SD_2}
\end{align}
where 
\begin{align}
{\mathcal L}_{SD} :=
\left\{
\frac{\eta}{2N}\sum_{i=1}^{N}\frac{\partial^{2}}{\partial E_{i}^{2}}
+\frac{\eta}{2N}
\sum_{l\neq i}^{N}
\frac{1}{E_{i}-E_{l}}\frac{\partial}{\partial E_{i}}-2\sum_{k=1}^{N}E_{k}\frac{\partial}{\partial E_{k}}-\frac{N(N+1)}{2}
\right\}. \label{LSD}
\end{align}
\end{theorem}

\section{Diagonalization of ${\mathcal L}_{SD}$}\label{sect3}
In this section we prove the main theorem (Theorem \ref{main_thm1}). 
The calculations in this section are performed in the similar manner as the calculations in \cite{Grosse:2023jcb}; we refer to \cite{Grosse:2023jcb}
for further details.

\bigskip

As the first step we prove the following proposition.
\begin{proposition}\label{prop3_1}
The differential operator ${\mathcal L}_{SD} $
defined in (\ref{LSD}) is transformed as
\begin{align}
&e^{-\frac{N}{\eta}\sum_{i=1}^{N}E_{i}^{2}}\Delta(E)^{\frac{1}{2}}\mathcal{L}_{SD}\Delta(E)^{-\frac{1}{2}}e^{\frac{N}{\eta}\sum_{i=1}^{N}E_{i}^{2}}
= - {\mathcal H}_{CM} .
\end{align}
Here, we denote the Hamiltonian of the Calogero-Moser model by ${\mathcal H}_{CM}$:
\begin{align}\label{hamiltonian_ho}
{\mathcal H}_{CM}:=-\frac{\eta}{2N}\left(\sum_{i=1}^{N}\frac{\partial^{2}}{\partial E_{i}^{2}}+\frac{1}{4}\sum_{i\neq j}\frac{1}{(E_{i}-E_{j})^{2}}\right)+2\frac{N}{\eta}\sum_{i=1}^{N}E_{i}^{2}.
\end{align}
\end{proposition}
\begin{proof}
By direct calculations,
we obtain
\begin{align}
&\Delta(E)^{\frac{1}{2}}\left(\frac{\eta}{2N}\sum_{i=1}^{N}\frac{\partial^{2}}{\partial E_{i}^{2}}
+\frac{\eta}{2N}
\sum_{l\neq i}^{N}
\frac{1}{E_{i}-E_{l}}\frac{\partial}{\partial E_{i}}\right)\Delta(E)^{-\frac{1}{2}}
\notag \\
&=\frac{\eta}{2N}\sum_{i=1}^{N}\frac{\partial^{2}}{\partial E_{i}^{2}}
+\frac{\eta}{8N}
\sum_{l\neq i}^{N}
\frac{1}{(E_{i}-E_{l})^{2}}.
\end{align}
Here, we used $\displaystyle \sum_{i\neq l\neq k \neq i} \frac{1}{(E_i-E_l)(E_i-E_k)}=0$.
Next we calculate the following:
\begin{align}
&\Delta(E)^{\frac{1}{2}}\left(-2\sum_{k=1}^{N}E_{k}\frac{\partial}{\partial E_{k}}\right)\Delta(E)^{-\frac{1}{2}}
=
\sum_{l\neq k}^{N}
\frac{E_{k}}{E_{k}-E_{l}}-2\sum_{k=1}^{N}E_{k}\frac{\partial}{\partial E_{k}}\notag\\
&=\sum_{k>l}1-2\sum_{k=1}^{N}E_{k}\frac{\partial}{\partial E_{k}}
=\frac{N(N-1)}{2}-2\sum_{k=1}^{N}E_{k}\frac{\partial}{\partial E_{k}} .
\end{align}
Then, we obtain
\begin{align}
&\Delta(E)^{\frac{1}{2}}
{\mathcal L}_{SD}
\Delta(E)^{-\frac{1}{2}}
=&\frac{\eta}{2N}\Biggl\{\sum_{i=1}^{N}\frac{\partial^{2}}{\partial E_{i}^{2}}+\frac{1}{4}\sum_{i\neq j}\frac{1}{(E_{i}-E_{j})^{2}}\Biggl\}-2\sum_{k=1}^{N}E_{k}\frac{\partial}{\partial E_{k}}-N. \label{200}
\end{align}
Using 
\begin{align}
&  e^{-\frac{N}{\eta}\sum_{i=1}^{N}E_{i}^{2}}\left(\frac{\eta}{2N}\sum_{i=1}^{N}\left(\frac{\partial}{\partial E_{i}}\right)^{2}\right)e^{\frac{N}{\eta}\sum_{i=1}^{N}E_{i}^{2}}\nonumber
\\
  &=N+2\sum_{i=1}^{N}E_{i}\frac{\partial}{\partial E_{i}}+\frac{\eta}{2N}\sum_{i=1}^{N}\frac{\partial^{2}}{\partial E_{i}^{2}}+\frac{2N}{\eta}\sum_{i=1}^{N}E_{i}^{2}\label{201}
\end{align}
and
\begin{align}
e^{-\frac{N}{\eta}\sum_{i=1}^{N}E_{i}^{2}}\left(-2\sum_{k=1}^{N}E_{k}\frac{\partial}{\partial E_{k}}\right)e^{\frac{N}{\eta}\sum_{i=1}^{N}E_{i}^{2}}=&-4\frac{N}{\eta}\sum_{k=1}^{N}E_{k}^{2}-2\sum_{k=1}^{N}E_{k}\frac{\partial}{\partial E_{k}}, \label{202}
\end{align}
%
%
finally we obtain
\begin{align}
&e^{-\frac{N}{\eta}\sum_{i=1}^{N}E_{i}^{2}}\Delta(E)^{\frac{1}{2}}\mathcal{L}_{SD}\Delta(E)^{-\frac{1}{2}}e^{\frac{N}{\eta}\sum_{i=1}^{N}E_{i}^{2}}
= -\mathcal{H}_{CM} .
\end{align}

\end{proof}
\bigskip

We introduce a function $\Psi(E,\eta):=e^{-\frac{N}{\eta}\sum_{i=1}^{N}E_{i}^{2}}\Delta(E)^{\frac{1}{2}}{Z}(E,\eta)$, 
then we obtain $\mathcal{H}_{CM}\Psi(E,\eta)=0$ 
from Proposition \ref{prop3_1} and Theorem \ref{thm2_2}. 
Thus, the Theorem \ref{main_thm1} was proved.\\
\bigskip

The Hamiltonian of the Calogero-Moser model is defined as follows \cite{Calogero:1970nt,kakei}:
\begin{align}
\displaystyle{H}_{C_{\beta}}:=\frac{1}{2}\sum_{j=1}^{N}\left(-\frac{\partial^{2}}{\partial y_{j}^{2}}+y_{j}^{2}\right)
+\sum_{j>k}\frac{\beta(\beta-1)}{(y_{j}-y_{k})^{2}}.\label{CM_beta}
\end{align}
After changing variable $\displaystyle \sqrt{\frac{2N}{\eta}}E_{i}=y_{i}$, 
if $\displaystyle\beta=\frac{1}{2}$, 
(\ref{Calogero_H_for_Matrix}) is identified with (\ref{CM_beta})
up to global factor $\frac{1}{2}$:
\begin{align}
\displaystyle H_{C_{\beta=\frac{1}{2}}}=\frac{1}{2}\sum_{j=1}^{N}\left(-\frac{\partial^{2}}{\partial y_{j}^{2}}+y_{j}^{2}\right)
-\frac{1}{4}\sum_{j>k}\frac{1}{(y_{j}-y_{k})^{2}}=\frac{1}{2}\mathcal{H}_{CM}.
\end{align}
In the following, we consider only the case $\displaystyle\beta=\frac{1}{2}$.

\section{Virasoro algebra }

Bergshoeff and Vasiliev proved in \cite{Bergshoeff:1994dd} 
that the Calogero-Moser model is associated with a Virasoro algebra
structure. 
In this section, we discuss the Virasoro algebra
representation in our $\Phi^4$ real symmetric matrix model.

As a start, a variable transformation is performed so that the
Hamiltonian obtained in the previous section coincides with the
Hamiltonian of the one in \cite{Bergshoeff:1994dd}.

Using $\displaystyle y_{i}= \sqrt{\frac{2N}{\eta}}E_{i}$, ${\mathcal L}_{SD}$
is expressed as
\begin{align}
-\frac{1}{2}{\mathcal L}_{SD}
=&\sum_{k=1}^{N}y_{k}\frac{\partial}{\partial y_{k}}
-\frac{1}{2}\Biggl\{
\sum_{i=1}^{N}\frac{\partial^{2}}{\partial y_{i}^{2}}+\frac{1}{2}
\sum_{l\neq i}^{N}\frac{1}{y_{i}-y_{l}}\left(\frac{\partial}{\partial y_{i}}-\frac{\partial}{\partial y_{l}}\right)
\Biggl\}+\frac{N(N+1)}{4}. \label{aua}
\end{align}
As we saw in Section \ref{sect3}, the Hamiltonian 
of Calogero-Moser model with $\displaystyle\beta=\frac{1}{2}$ is given as
\begin{align}
H_{C_{\beta=\frac{1}{2}}}=&g \left(-\frac{1}{2}{\mathcal L}_{SD} \right)g^{-1}.
\label{eq4_2}
\end{align}
Here $\displaystyle g=e^{-\frac{1}{2}\sum_i y_i^2 }\prod_{j>k}(y_{j}-y_{k})^{\frac{1}{2}}$.

\subsection{Review of the Virasoro algebra symmetry representation for 
the Calogero-Moser model}

In this subsection, we review 
several results of
\cite{Bergshoeff:1994dd}.
As \cite{Bergshoeff:1994dd, kakei}
we define the creation, annihilation operators
$ a_i^\dagger , a_i$, and the coordinate swapping operator $K_{ij} \quad (i,j = 1,...,N)$
obeying the following relations:
\begin{eqnarray}
\label{eq:algebra}
[a_i , a_j ]&=&[a^{\dagger}_i ,a^{\dagger}_j ]=0, \quad 
[a_i ,a^{\dagger}_j ]
= A_{ij} := 
\delta_{ij }\left(1+\beta\sum_{l=1}^N K_{il}\right)-\beta K_{ij}, \\
K_{ij}K_{jl}&=&K_{jl}K_{il}=K_{il}K_{ij}, \quad
\mbox{ for  all $ i \ne j, i \ne l,  j \ne l $} , \\
(K_{ij})^2&=&I\,,\qquad K_{ij}=K_{ji} , \\
K_{ij}K_{mn} &=& K_{mn}K_{ij} , \quad 
\mbox{ if all indices $i,j,m,n$ are different}, \\
K_{ij}a^{(\dagger)}_j&=&a^{(\dagger)}_i K_{ij}.
\end{eqnarray}
Here,
we chose 
$\displaystyle\beta=\frac{1}{2}$ for our case, while
$K_{ij}$  are the elementary permutation
operators of the symmetric group $\mathfrak{S}_N$. 
$K_{ij}$ means the replacement of coordinates as $K_{ij}y_i = y_j$
in the following discussions. 
We use the standard convention
that square brackets $[ *, * ]$ denote commutators and 
curly brackets $\{ * , * \}$ anticommutators.

To make contact with the
Calogero-Moser model, we chose these operators as
\begin{equation} \label{a_ad_y_D}
a_i = \frac{1}{\sqrt{2}} (y_i + D_i)\,, \quad
a^\dagger_i = \frac{1}{\sqrt{2}} (y_i - D_i)\,,
\end{equation}
with Dunkl derivatives \cite{Du1,kakei} 
\begin{equation}
D_i =\frac{\partial }{\partial y_i}+\beta \sum_{j=1, j\neq i}^N (y_i -y_j)^{-1}
(1-K_{ij})\, .
\end{equation}
We can show it by direct calculations that 
the coordinates and the Dunkl derivatives
satisfy the following commutation
relations \cite{Br1,Po1}:
\begin{equation}
[y_i, y_j] = [D_i, D_j] = 0, \hskip .3truecm [D_i, y_j] = A_{ij},
\end{equation}
and then we find that the relations (\ref{eq:algebra}) are also satisfied by 
(\ref{a_ad_y_D})\cite{Bergshoeff:1994dd}.\\
\bigskip

Let us introduce the following Hamiltonian like a harmonic oscillator system:
\begin{equation}
\label{eq:ham}
H = {1\over 2} \sum_{i=1}^N \{a_i \,,a^\dagger_i\} \,.
\end{equation}
This Hamiltonian and
$H_{C_{\beta=\frac{1}{2}}}$ are related as
\begin{align}
\mathrm{Res}(H)=&\prod_{j>k}(y_{j}-y_{k})^{-\frac{1}{2}}\cdot H_{C_{\beta=\frac{1}{2}}}\cdot \prod_{j>k}(y_{j}-y_{k})^{\frac{1}{2}}\notag\\
=&\frac{1}{2}\sum_{j=1}^{N}\left(-\frac{\partial^{2}}{\partial y_{j}^{2}}+y_{j}^{2}\right)-\frac{1}{4}\sum_{j\neq k}\frac{1}{y_{j}-y_{k}}\left(\frac{\partial}{\partial y_{j}}-\frac{\partial}{\partial y_{k}}\right),
\end{align}
where $\displaystyle\mathrm{Res}(H)$ means that
operator $H$ acts on
symmetric function space.
It is possible to represent any differential operator 
$D$ including $K_{ij}$'s
as placing the elements of $S_n$
at the right end, i.e. $D= \sum_{\omega\in S_{N}}D_{\omega}\omega $.
Using this expression,  $\displaystyle\mathrm{Res}$ is defined as
$\displaystyle\mathrm{Res}\left(
\sum_{\omega\in S_{N}}D_{\omega}\omega
\right)=\sum_{\omega\in S_{N}}D_{\omega}$. 
The Hamiltonian satisfies 
\begin{equation}
[ H ~, ~ a^\dagger_i ] = a^\dagger_i , \qquad [ H~, ~ a_i ] = -a_i 
\end{equation}
as well as the harmonic oscillator case.
Next we define the representation of the Virasoro generators 
using Dunkl operators:
\begin{equation}
\label{eq:Ansatz}
L_{-n} = \sum_{i=1}^N \left ( \alpha (a_i^\dagger)^{n+1} a_i +
(1-\alpha) a_i (a_i^\dagger)^{n+1}
+\left(\lambda-{1\over 2}\right) (n+1) (a_i^\dagger)^{n} \right ) ,
\end{equation}
where $\alpha, \lambda$ are arbitrary parameters. 
Or more generally, for any Laurent series $\xi (a_i^\dagger) $,
we can define the Virasoro generators by
\begin{equation}
L_{\xi} = \sum_{i=1}^N \left ( \alpha 
\xi (a_i^\dagger) a_i +
(1-\alpha) a_i  \xi (a_i^\dagger)
+\left(\lambda-{1\over 2}\right)  
\frac{\partial}{\partial a_i^\dagger}    \xi(a_i^\dagger) \right ) .
\end{equation}
For simplicity, 
we chose
$\displaystyle\lambda=\frac{1}{2}$ in this paper, however, 
this choice is not essential
in the following discussion. 
When $\xi_1$ and $\xi_2$ are arbitrary Laurent series, $[L_{\xi_1}, L_{\xi_2}]$ is as follows:
\begin{equation}
[L_{\xi_1}, L_{\xi_2}] =
\sum_{i=1}^N\left ( \alpha \xi_{1,2}(a_i^\dagger)a_i +
(1-\alpha) a_i \xi_{1,2}(a_i^\dagger)\right ) ,
\end{equation}
where $\xi_{1,2}(a_i^\dagger)$ is defined by
\begin{equation}
\xi_{1,2}(a_i^\dagger) =
\xi_1(a_i^\dagger) {\partial\over\partial a_i^\dagger}
\xi_2(a_i^\dagger)
-\xi_2(a_i^\dagger) {\partial\over\partial a_i^\dagger}
\xi_1(a_i^\dagger) .
\end{equation}
Especially if $\displaystyle L_{-n} = \sum_{i=1}^N \left ( \alpha (a_i^\dagger)^{n+1} a_i +
(1-\alpha) a_i (a_i^\dagger)^{n+1}\right )$, 
their commutators 
are given by the ones of the Virasoro algebra
with its central charge $c=0$:
\begin{align}
\displaystyle[L_{n},L_{m}]=(n-m)L_{n+m}. 
\end{align}

\subsection{Virasoro algebra representation for real symmetric $\Phi^4$-matrix model}
We shall attempt to adapt the Virasoro algebra reviewed in the previous subsection to the matrix model we are considering.\\

From $\displaystyle H=L_{0}-\biggl(\frac{1}{2}-\alpha\biggl)N+\frac{1}{2}\biggl(\alpha-\frac{1}{2}\biggl)\sum_{i\neq j}K_{ij}$,
the commutator $[H ~, ~L_{-m}]$ is obtained as
\begin{align}
[H,L_{-m}]=&mL_{-m}+\left[\frac{1}{2}\Biggl(\alpha-\frac{1}{2}\Biggl)\sum_{i\neq j}K_{ij}~, ~ \sum_{i=1}^{N}\Biggl(\alpha(a_{i}^{\dagger})^{m+1}a_{i}+(1-\alpha)a_{i}(a_{i}^{\dagger})^{m+1}\Biggl)\right].\label{kinoko}
\end{align}
Let us calculate $[K_{pq}~, ~\sum_{i=1}^{N}(a_{i}^{\dagger})^{m}(a_{i})^{n} ]$.
When $p\neq q$,
\begin{align}
&\left[K_{pq}~, ~\sum_{i=1}^{N}(a_{i}^{\dagger})^{m}(a_{i})^{n}\right] \label{4_21}  \\
&=\sum_{i\neq q, i\neq p}\left(K_{pq}(a_{i}^{\dagger})^{m}(a_{i})^{n}-(a_{i}^{\dagger})^{m}(a_{i})^{n}K_{pq}\right)\notag\\
&\mbox{ \hspace{5mm} } +\Biggl(K_{pq}(a_{p}^{\dagger})^{m}(a_{p})^{n}-(a_{p}^{\dagger})^{m}(a_{p})^{n}K_{pq}
+K_{pq}(a_{q}^{\dagger})^{m}(a_{q})^{n}-(a_{q}^{\dagger})^{m}(a_{q})^{n}K_{pq}\Biggl)\notag\\
&=\Biggl((a_{q}^{\dagger})^{m}(a_{q})^{n}K_{pq}-(a_{p}^{\dagger})^{m}(a_{p})^{n}K_{pq}+(a_{p}^{\dagger})^{m}(a_{p})^{n}K_{pq}-(a_{q}^{\dagger})^{m}(a_{q})^{n}K_{pq}\Biggl)
=0. \notag
\end{align}
When $p=q$, 
$[K_{pp} ~, ~\sum_{i=1}^{N}(a_{i}^{\dagger})^{m}(a_{i})^{n}]=0$
is trivial. 
For any $p,q$,
$
\left[K_{pq} , 
\sum_{i=1}^{N} (a_{i})^{m}(a_{i}^{\dagger})^{n}
\right]=0 $ is calculated similarly. 
From these results, (\ref{kinoko}) is simplified as
\begin{align}
[H~,~L_{-m}]
=&m L_{-m}.
\end{align}
From (\ref{eq4_2}),
\begin{align}
-\frac{1}{2} \mathcal{L}_{SD} = 
e^{\frac{1}{2}\sum_j y_j^2} \mathrm{Res} ( H ) 
e^{-\frac{1}{2}\sum_j y_j^2} .
\end{align}
Note that the functions $e^{-\frac{1}{2}\sum_j y_j^2} $, $e^{\frac{1}{2}\sum_j y_j^2} $, and
the partition function $Z(E , \eta)$ are invariants under $\mathfrak{S}_N$ action, i.e.\ $K_{ij} Z(E , \eta) = Z(E , \eta)$, and so on,
so that we can ignore $\mathrm{Res} $ in the following calculations.
Let us introduce
$\displaystyle\widetilde{L}_{-m}:=
e^{\frac{1}{2}\sum_j y_j^2} L_{-m}
e^{-\frac{1}{2}\sum_j y_j^2} .
$
The following is automatically satisfied:
\begin{align}
[\widetilde{L}_{n} ~ ,~ \widetilde{L}_{m}]=(n-m) \widetilde{L}_{n+m}. 
\end{align}
More explicitly, using 
\begin{align}
&e^{\frac{1}{2}\sum_j y_j^2} D_i 
e^{-\frac{1}{2}\sum_j y_j^2} =
D_i -y_i ,  \\
&\widetilde{a}_{i} :=e^{\frac{1}{2}\sum_j y_j^2} a_i 
e^{-\frac{1}{2}\sum_j y_j^2} = \frac{1}{\sqrt{2}} D_i \\
&\widetilde{a}_{i}^{\dagger} :=e^{\frac{1}{2}\sum_j y_j^2} a_i^{\dagger} 
e^{-\frac{1}{2}\sum_j y_j^2} = \frac{1}{\sqrt{2}} (2y_i - D_i) ,
\end{align}
$\widetilde{L}_{-n} $
is expressed as
\begin{align}
\widetilde{L}_{-n} &= \sum_{i=1}^N \left ( \alpha (\widetilde{a}_i^\dagger)^{n+1} \widetilde{a}_i +
(1-\alpha) \widetilde{a}_i (\widetilde{a}_i^\dagger)^{n+1}
\right ) \notag\\
&= \frac{1}{2^{(n+2)/2}}\sum_{i=1}^{N}\Biggl\{\alpha\left(-D_{i}+2y_{i}\right)^{n+1}D_{i}+(1-\alpha)D_{i}\left(-D_{i}+2y_{i}\right)^{n+1}\Biggl\}.
\end{align}
It is better to rewrite these operators 
using the original matrix model variables, $E_i$ and $\eta$.
Let us introduce
\begin{align*}
D_i^E := \frac{\partial }{\partial E_i}
+\frac{1}{2} \sum_{j=1, j\neq i}^N \frac{1}{(E_i - E_j)}
(1-K_{ij})
= \sqrt{ \frac{2N}{\eta}} D_i .
\end{align*}
Of course, this operator $D_i^E$ satisfies 
$[ D_i^E , E_j] = A_{ij}$ and $[ D_i^E , D_j^E]=0$. 
Using this $D_i^E$, the operators
$\widetilde{a}_{i} , \widetilde{a}_{i}^{\dagger}$
and $\widetilde{L}_{-n}$ are written as
\begin{align}
\widetilde{a}_{i} =& \frac{1}{2} \sqrt{\frac{\eta}{N}} D_i^E ,
\qquad \widetilde{a}_{i}^{\dagger} = 2 \sqrt{\frac{N}{\eta}} E_i 
- \frac{1}{2}\sqrt{\frac{\eta}{N}}D_i^E ,\\
\widetilde{L}_{-n} =&
\sum_{i=1}^N \left\{ \alpha \left(2 \sqrt{\frac{N}{\eta}} E_i 
- \frac{1}{2}\sqrt{\frac{\eta}{N}}D_i^E\right)^{n+1}  \frac{1}{2} \sqrt{\frac{\eta}{N}} D_i^E \right. \notag \\
&+\left.
(1-\alpha)  \frac{1}{2} \sqrt{\frac{\eta}{N}} D_i^E
\left(2 \sqrt{\frac{N}{\eta}} E_i 
- \frac{1}{2}\sqrt{\frac{\eta}{N}}D_i^E \right)^{n+1}
\right\}.
\end{align} 
\bigskip

Recall $\mathcal{L}_{SD} = -2 e^{\frac{1}{2}\sum_j y_j^2} \mathrm{Res}(H)
e^{-\frac{1}{2}\sum_j y_j^2} $ and  (\ref{4_21}), then
\begin{align}
\left[\mathcal{L}_{SD} ~, ~\widetilde{L}_{-m}\right]
=& -2 e^{\frac{1}{2}\sum_j y_j^2}  [  \mathrm{Res}(H) ~ , ~ L_{-m} ] e^{-\frac{1}{2}\sum_j y_j^2} 
 \notag \\
=& -2 e^{\frac{1}{2}\sum_j y_j^2}  [  L_0 ~ , ~ L_{-m} ] e^{-\frac{1}{2}\sum_j y_j^2}
=-2 m\widetilde{L}_{-m} . \label{4_30}
\end{align}
From Theorem \ref{thm2_2} and (\ref{4_30}),
finally we get the following theorem.
\begin{theorem}\label{main2}
The partition function defined by (\ref{partitionfunction}) satisfies
\begin{align}
\mathcal{L}_{SD}(\widetilde{L}_{-m}{Z}(E,\eta))=&-2m(\widetilde{L}_{-m}{Z}(E,\eta)).
\end{align}
This means that $\widetilde{L}_{-m}{Z}(E,\eta)$
is an eigenfunction of $\mathcal{L}_{SD}$ with 
the eigenvalue $-2m$.
\end{theorem}

\bigskip

\noindent
{\bf Acknowledgment}\\
{
A.S.\ was supported by JSPS KAKENHI Grant Number 21K03258. R.W.\ was
supported\footnote{``Funded by
  the Deutsche Forschungsgemeinschaft (DFG, German Research
  Foundation) -- Project-ID 427320536 -- SFB 1442, as well as under
  Germany's Excellence Strategy EXC 2044 390685587, Mathematics
  M\"unster: Dynamics -- Geometry -- Structure.''} by the Cluster of
Excellence \emph{Mathematics M\"unster} and the CRC 1442 \emph{Geometry:
  Deformations and Rigidity}.
This study was supported by Erwin Schr\"odinger International Institute for Mathematics and Physics (ESI) through the project
\emph{Research in Teams} Project ``Integrability''.
}\\

\noindent
{\bf Data availability} \  No datasets were generated or analyzed
during the current study.

\section*{Declarations}

{\bf Conflicts of interest} \ On behalf of all authors, the corresponding author states that there is no conflict of
interest.

\appendix

\section{Appendix A} \label{appendixA}

We give the proof 
for (\ref{formulaB})
in this Appendix \ref{appendixA}.
(The first half of this proof consists of well-known facts.
For example, (\ref{ap6}) can be seen in \cite{Petersen}. 
However, for the reader's convenience, the derivation of equation (\ref{ap6}) has not been omitted.)

\begin{proof}
For a real symmetric matrix $X={}^T\!X=(x_{ij})$, $\displaystyle\frac{\partial X}{\partial x_{ij}}=E_{ij}+E_{ji}-E_{ij}E_{ji}\delta_{ij}=E_{ij}+E_{ji}-E_{ij}E_{ij}$,
where $E_{ij}$ is standard matrix basis with $1$ on $ij$ position, i.e. $E_{ij}=(\delta_{ki}\delta_{lj})$. 
Or, equivalently it is written as $\frac{\partial x_{kl}}{\partial x_{ij}}=\delta_{ki}\delta_{jl}+\delta_{kj}\delta_{il}-\delta_{ki}\delta_{jl}\delta_{ij}$.
Then,
\begin{align}
  \mathrm{Tr}\left(X^{-1}\frac{\partial X}{\partial x_{ij}}\right)
  &=\mathrm{Tr}\left(X^{-1}\left(E_{ij}+E_{ji}-E_{ij}E_{ij}\right)\right)\notag\\
&=(X^{-1})_{ji}+(X^{-1})_{ij}-(X^{-1})_{ji}\delta_{ij} 
= 2(X^{-1})_{ij}-(X^{-1})_{ij}\delta_{ij}  \label{300}
\end{align}
since $X$ is symmetric.
Next we calculate 
\begin{align}
\frac{\partial\det(X)}{\partial x_{ij}}=&\frac{\partial}{\partial x_{ij}} \exp{( \mathrm{Tr}\log X )}
=\mathrm{Tr} \left(\frac{\partial\log X}{\partial x_{ij}} \right) \det X. \label{100}
\end{align}
By partial differentiation of 
$\mathrm{Tr}\left(X^{-1}\left(\exp\left(\log X\right)\right)\right)
=\mathrm{Tr} (\mathrm{Id}) $ with respect to $x_{ij}$ we obtain
\begin{align}
\mathrm{Tr}\frac{\partial\log X}{\partial x_{ij}}=-\mathrm{Tr} \left(\frac{\partial X^{-1}}{\partial x_{ij}}X\right)=\mathrm{Tr} 
\left(X^{-1}\frac{\partial X}{\partial x_{ij}} \right).
\label{101}
\end{align}
From (\ref{101}) and (\ref{100}), we find
\begin{align}
\frac{\partial\det(X)}{\partial x_{ij}}
&=\mathrm{Tr}\left(X^{-1}\frac{\partial X}{\partial x_{ij}}\right)\det (X) . \label{301}
\end{align}
Substituting (\ref{300}) into (\ref{301}), we get
\begin{align}
\frac{\partial\det(X)}{\partial x_{ij}}
=&\left(2(X^{-1})_{ij}-(X^{-1})_{ij}\delta_{ij}\right)\det (X) .
\label{ap6}
\end{align}
%
We define $\widetilde{B}$ as the cofactor matrix of  $B$.  Applying (\ref{ap6}) for $P(x)=\det (B)$,
%
\begin{align}
\frac{\partial P(x)}{\partial B_{ij}}=&-\frac{\partial P(x)}{\partial H_{ij}}=P(x)\Biggl\{+2(B^{-1}(x))_{ij}-(B^{-1}(x))_{ii}\delta_{ij}\Biggl\}\notag\\
=&2({}^T\!\widetilde{B}(x))_{ij}-({}^T\!\widetilde{B}(x))_{ii}\delta_{ij} . \label{102}
\end{align}
On the other hand from $P(x)=\prod_{i=1}^{N}(x-E_{i})$,
\begin{align}
\frac{\partial P(x)}{\partial H_{ij}}=&\sum_{l=1}^{N}\frac{\partial E_{l}}{\partial H_{ij}}\frac{\partial}{\partial E_{l}}\prod_{i=1}^{N}(x-E_{i})
=\sum_{l=1}^{N}\frac{\partial E_{l}}{\partial H_{ij}}\frac{-P(x)}{x-E_{l}}.
\label{103}
\end{align}
From (\ref{102}) and (\ref{103}),
\begin{align}
-2(^T\!\widetilde{B}(x))_{ij}+(^T\!\widetilde{B}(x))_{ii}\delta_{ij}=&\sum_{l=1}^{N}\frac{\partial E_{l}}{\partial H_{ij}}\frac{-P(x)}{x-E_{l}} 
\end{align}
is obtained. Setting $x= E_t$,
\begin{align}
-2(^T\!\widetilde{B}(E_t))_{ij}+(^T\!\widetilde{B}(E_t))_{ii}\delta_{ij}
=&-\sum_{l=1}^{N}\frac{\partial E_{l}}{\partial H_{ij}}\prod_{k=1,\hspace{2mm}k\neq l}^{N}(E_{t}-E_{k})\notag\\
=&-\sum_{l=1}^{N}\frac{\partial E_{l}}{\partial H_{ij}}\delta_{tl}P'(E_{t})
=- \frac{\partial E_{t}}{\partial H_{ij}}P'(E_{t}) .
\label{104}
\end{align}
From (\ref{104}),
finally we get the result we want:
\begin{align}
\frac{\partial E_{t}}{\partial H_{ij}}=&\frac{2(^T\!\widetilde{B}(E_{t}))_{ij}-(^T\!\widetilde{B}(E_{t}))_{ii}\delta_{ij}}{P'(E_{t})}.
\label{105}
\end{align}

\end{proof}


\end{document}